%% file: main.tex
\documentclass{article}

\usepackage{iclr2021_conference}
\usepackage{times}
\usepackage[utf8]{inputenc}
\usepackage{pgfplots}
\usepackage{hyperref}
\usepackage{url}
\usepackage{amsmath}
\usepackage{amssymb}
\usepackage{amsthm}

\pgfplotsset{compat=1.16}
\pgfplotsset{
    tick label style={font=\sffamily\footnotesize},
    every axis label={font=\sffamily\small},
    legend style={font=\sffamily\small},
    label style={font=\sffamily\small}
}
\usetikzlibrary{arrows}

\definecolor{vqcolor}{RGB}{200, 120, 20}
\definecolor{uqcolor}{RGB}{49, 130, 189}

\newtheorem{theorem}{Theorem}

\title{On the advantages of stochastic encoders}

\author{Lucas Theis\thanks{Equal contribution} \\
Google Research \\
\texttt{theis@google.com}
\And 
Eirikur Agustsson\footnotemark[1] \\
Google Research \\
\texttt{eirikur@google.com}}

\iclrfinalcopy

\input{macros.tex}

\begin{document}

\maketitle

\begin{abstract}
Stochastic encoders have been used in rate-distortion theory and neural compression because they can be easier to handle. However, in performance comparisons with deterministic encoders they often do worse, suggesting that noise in the encoding process may generally be a bad idea. It is poorly understood if and when stochastic encoders can do better than deterministic encoders. In this paper we provide one illustrative example which shows that stochastic encoders can significantly outperform the best deterministic encoders. Our toy example suggests that stochastic encoders may be particularly useful in the regime of ``perfect perceptual quality''.
\end{abstract}

\section{Introduction}
In lossy compression, information is typically deterministically encoded through the following procedure. First, the input is transformed using a fixed analysis transform or a learned neural network. This is followed by a quantization step mapping continuous to discrete values. Finally, the remaining information is further compressed and encoded into bits using entropy encoding. Taken together, these steps implement a deterministic function from the input to the integers.
In this paper we discuss whether it is ever advantageous to stochastically encode data by introducing noise into this process.

To make this discussion more precise, let
\begin{align}
    \mathcal{F} = \{ f: \mathbb{R}^n \times \mathbb{R} \rightarrow \mathbb{N}_0 \}, \quad
    \mathcal{G} = \{ g: \mathbb{N}_0 \times \mathbb{R} \rightarrow \mathbb{R}^n \}
\end{align}
be sets of \textit{stochastic encoders} and \textit{decoders}. The encoder expects $n$-dimensional data as well as an additional source of randomness as input. This source of randomness will be represented by a random variable distributed uniformly between zero and one, which we can think of as an infinite number of random bits when expressed in binary form ($U = 0.B_1B_2B_3...$). The decoder in turn expects the output of an encoder and a source of randomness as input. This could either be the same random input as for the encoder or a different source (Fig.~\ref{fig:graphs}). By deterministic encoders we mean
\begin{align}
    \mathcal{\tilde F} = \{ f \in \mathcal{F} \mid \exists \tilde f: \forall \x, u: f(\x, u) = \tilde f(\x) \},
\end{align} 
that is, encoders which ignore the random inputs. Defined in this way, it is clear that stochastic encoders can perform no worse than deterministic encoders since $\smash{\mathcal{\tilde F} \subset \mathcal{F}}$. However, for a given loss and data distribution we can ask whether an optimal encoder can be found in $\smash{\mathcal{\tilde F}}$ or whether we should hope to do better by considering the larger set $\mathcal{F}$.

Rate-distortion theory considers stochastic encoders because they can be more amenable to a theoretical analysis. For example, \textit{universal quantization} is a randomized quantization procedure whose error distribution is independent of the source \citep{ziv1985universal}, making it easier to analyze than deterministic quantization. In machine learning, stochastic encoders recently have been considered because they allow for the transmission of continuous signals and thus enable differentiable losses \citep{havasi2018miracle}. While this can make it easier to optimize stochastic encoders, the introduction of noise can also hurt performance \citep{agustsson2020uq}.

For the Laplace and exponential distributions, optimal quantizers exist which are deterministic \citep{sullivan1996laplace}. Universal quantization is known to introduce redundancy compared to the best deterministic quantizer at the same squared error distortion level \citep{ziv1985universal}. More generally we can say that the Lagrangian of a rate-distortion problem is always minimized by a deterministic encoder (Appendix~A). This begs the question whether stochastic encoders are ever better than deterministic encoders. \cite{wagner2021sawbridge} provide one example where deterministic encoders are only optimal at certain levels of distortion. For levels in between, a better rate can be achieved by randomly choosing one of two deterministic codecs, proving that stochastic encoders are indeed sometimes needed for optimal performance. Here we consider another setting where stochastic encoders appear to provide further advantages.

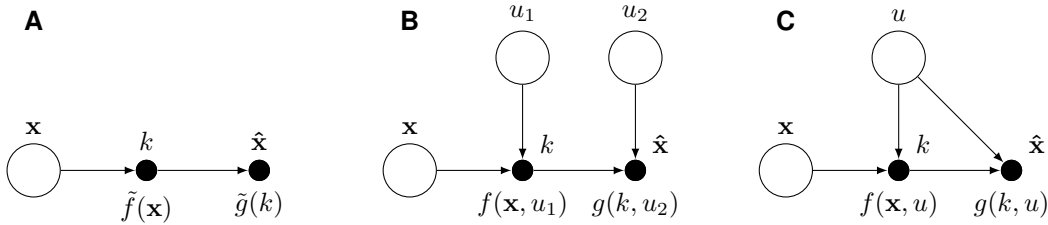
\begin{figure}
    \centering
    \begin{tikzpicture}
        \node at (0, 2) {\textbf{\textsf{A}}};
        \begin{scope}
            \node[fill, circle, inner sep=0.1cm, label=90:{$\hatx$}] (r) at (3, 0) {};
            \node[fill, circle, inner sep=0.1cm, label=270:{$\tilde g(k)$}] (r) at (3, 0) {};
            \node[fill, circle, inner sep=0.1cm, label=90:{$k$}] (z) at (1.5, 0) {};
            \node[fill, circle, inner sep=0.1cm, label=270:{$\tilde f(\x)$}] (z) at (1.5, 0) {};
            \node[draw, circle, inner sep=0.25cm, label=90:{$\mathbf{x}$}] (x) at (0.0, 0) {};
            \draw[-latex] (x) -- (z);
            \draw[-latex] (z) -- (r);
        \end{scope}
        
        \node at (5, 2) {\textbf{\textsf{B}}};
        \begin{scope}[xshift=5cm]
            \node[fill, circle, inner sep=0.1cm, label=45:{$\hatx$}] (r) at (3, 0) {};
            \node[fill, circle, inner sep=0.1cm, label=270:{$g(k, u_2)$}] at (3, 0) {};
            \node[fill, circle, inner sep=0.1cm, label=270:{$f(\x, u_1)$}] (z) at (1.5, 0) {};
            \node[fill, circle, inner sep=0.1cm, label=45:{$k$}] at (1.5, 0) {};
            \node[draw, circle, inner sep=0.25cm, label=90:{$\x$}] (x) at (0.0, 0) {};
            \node[draw, circle, inner sep=0.25cm, label=90:{$u_1$}] (u1) at (1.5, 1.5) {};
            \node[draw, circle, inner sep=0.25cm, label=90:{$u_2$}] (u2) at (3, 1.5) {};
            \draw[-latex] (x) -- (z);
            \draw[-latex] (z) -- (r);
            \draw[-latex] (u1) -- (z);
            \draw[-latex] (u2) -- (r);
        \end{scope}
        
        \node at (10, 2) {\textbf{\textsf{C}}};
        \begin{scope}[xshift=10cm]
            \node[fill, circle, inner sep=0.1cm, label=270:{$g(k, u)$}] (r) at (3, 0) {};
            \node[fill, circle, inner sep=0.1cm, label=45:{$\hatx$}] at (3, 0) {};
            \node[fill, circle, inner sep=0.1cm, label=270:{$f(\x, u)$}] (z) at (1.5, 0) {};
            \node[fill, circle, inner sep=0.1cm, label=45:{$k$}] at (1.5, 0) {};
            \node[draw, circle, inner sep=0.25cm, label=90:{$\x$}] (x) at (0.0, 0) {};
            \node[draw, circle, inner sep=0.25cm, label=90:{$u$}] (u) at (1.5, 1.5) {};
            \draw[-latex] (x) -- (z);
            \draw[-latex] (z) -- (r);
            \draw[-latex] (u) -- (z);
            \draw[-latex] (u) -- (r);
        \end{scope}
    \end{tikzpicture}
    \caption{Visualizations of deterministic and stochastic encoders and decoders. Empty circles represent sources of randomness, filled circles represent variables which depend deterministically on their inputs. \textbf{A}:~A deterministic encoder and a deterministic decoder. \textbf{B}:~A stochastic encoder and a stochastic decoder with independent sources of randomness. \textbf{C}:~A stochastic encoder and a stochastic decoder with a shared
    source of randomness. Setting B generalizes A and C generalizes B. (A single uniform random variable can be split into two independent variables. For example, the encoder can use the odd numbered bits and the decoder the even numbered bits.)}
    \label{fig:graphs}
\end{figure}

\section{Perfect perceptual quality}
The setting we consider here is the setting of ``perfect perceptual quality'' \citep{blau2018tradeoff}.
We still wish to optimize a rate-distortion trade-off but under the constraint that reconstructions follow the data distribution, that is, the reconstructions are guaranteed to be free of any artefacts.
More formally, let $\hatX = g(f(\X, U), U)$. We seek to
\begin{align}
    \text{minimize} \quad
    \mathbb{E}[d(\X, \hatX)]
    \quad \text{w.r.t.} \quad
    f \in \mathcal{F}_R, \, g \in \mathcal{G}
    \quad \text{s.t.} \quad 
    \hatX \sim \X.
\end{align}
For simplicity, we will only consider fixed-rate encoders
\begin{align}
    \mathcal{F}_R = \{ f: \mathbb{R}^n \times \mathbb{R} \rightarrow \{ 0, \dots, 2^{R} - 1 \} \}.
\end{align}
\cite{tschannen2018dist} also called this problem ``distribution-preserving lossy compression''. Their approach combined deterministic encoders with stochastic decoders. In the next section we demonstrate how also making the encoder stochastic can be advantageous in this setting.

\section{Circular uniform distribution}
Consider a uniform distribution over a unit circle. It will be convenient to use polar coordinates $(\theta, r)$ in addition to Cartesian coordinates $(x_1, x_2)$. In polar coordinates, we can write down the marginal density
\begin{align}
    p(\theta) = \frac{1}{2\pi}.
\end{align}
By our requirement of perfect perceptual quality, any reconstruction $\hatx$ will have to live on the unit circle as well, that is, $\hat r = r = 1$. Any distortion can therefore be expressed as a function of $\theta$ and $\hat\theta$ and we choose
\begin{align}
    d(\theta, \hat\theta)
    = 1 - \cos(\hat\theta - \theta)
    = 1 - \hatx^\top\x
    = \frac{1}{2} \|\hatx - \x\|^2,
\end{align}
which is the cosine distance or equivalently the mean squared error. We further target a bit rate of 1, that is, $f(\x, u) \in \{0, 1\}$.

We will now compare two approaches for this example, vector quantization (which uses a deterministic encoder) and universal quantization \citep{ziv1985universal} (which uses a stochastic encoder).

\subsection{Vector quantization}
\label{sec:vq}
Choose $\theta_0 = 0$ and $\theta_1 = \pi$.  Let $f \in \mathcal{\tilde F}_1$ be the encoder with
\begin{align}
    f(\x, u) = \underset{k \in \{0, 1\}}{\text{argmin}}\, d(\theta, \theta_k)
\end{align}
for any input $\x$ with polar coordinates $(\theta, 1)$. To satisfy the perceptual quality constraint, we decode by uniformly sampling from the Voronoi section corresponding to $k$,
\begin{align}
    g(k, u) = \begin{pmatrix}\cos\hat\theta \\ \sin\hat\theta\end{pmatrix} 
    \quad\text{where}\quad
    \hat\theta = \theta_k + \pi ( u - \frac{1}{2} ).
\end{align}
For a fixed $\theta \in [-\pi/2, \pi/2)$, the error is
\begin{align}
    \mathbb{E}[d(\theta, \hat\Theta) \mid \theta]
    = 1 - \int_{\pi/2}^{-\pi/2} \frac{1}{\pi} \cos(\theta - \hat \theta) \, d\hat \theta
    = 1 - \frac{2}{\pi} \cos(\theta),
\end{align}
and similarly for other angles. The expected distortion is
\begin{align}
    \mathbb{E}[d(\Theta, \hat\Theta)] = 1 - \frac{4}{\pi^2}.
\end{align}
Appendix~C shows that this is the best performance we can hope for with a deterministic encoder.

\begin{figure}
    \centering
    \begin{tikzpicture}
        \node at (0, \textwidth/2.7) {\textbf{\textsf{A}}};
        \begin{axis}[
                width=\textwidth/2,
                height=\textwidth/2,
                xmin=-1.3,
                xmax=1.3,
                ymin=-1.3,
                ymax=1.3,
                ticks=none,
                axis lines=middle
            ]
            \draw[vqcolor, line width=3pt] (axis cs:1,0) arc [radius=1, start angle=0, end angle=180];
            \draw[black, line width=1pt] (axis cs:1,0) arc [radius=1, start angle=0, end angle=-180];
            
            \node[fill, black, circle, inner sep=2.5pt, label=20:{$\mathbf{x}$}] (x) at (axis cs:0.939693, 0.34202) {};

            \node[draw, vqcolor, circle, line width=1pt, inner sep=2.5pt, label=0:{$\mathbf{c}_1$}] at (axis cs:0, 0.63662) {};
            \node[draw, black, circle, inner sep=2.5pt, label=0:{$\mathbf{c}_2$}] at (axis cs:0, -0.63662) {};
        \end{axis}

        \node at (\textwidth/2, \textwidth/2.7) {\textbf{\textsf{B}}};
        \begin{axis}[
                xshift=\textwidth/2,
                width=\textwidth/2,
                height=\textwidth/2,
                xmin=-1.3,
                xmax=1.3,
                ymin=-1.3,
                ymax=1.3,
                ticks=none,
                axis lines=middle
            ]
            \draw[uqcolor, line width=3pt] (axis cs:0.34202, -0.939693) arc [radius=1, start angle=-70, end angle=110];
            \draw[black, line width=1pt] (axis cs:0.34202, -0.939693) arc [radius=1, start angle=-70, end angle=-250];
            
            \node[fill, black, circle, inner sep=2.5pt, label=20:{$\mathbf{x}$}] (x) at (axis cs:0.939693, 0.34202) {};
    
            \draw[black, line width=0.5pt] (axis cs:0.6,0) arc [radius=0.6, start angle=0, end angle=20];
            \draw (axis cs:0, 0) -- (x);
            \node at (axis cs:0.492404, 0.0868241) {$\theta$};
        \end{axis}
    \end{tikzpicture}
    \caption{\textbf{A}: Visualization of deterministic 1-bit encoding of data drawn uniformly from the circle. Vector quantization maps data points of a Voronoi region (orange) to a codebook vector ($\mathbf{c}_1$). Alternatively, a stochastic decoder may sample uniformly from the Voronoi region to achieve \textit{perfect perceptual quality}. Depending on the location of the data point, the decoder will sample from one of two regions. \textbf{B}: Universal quantization stochastically encodes data such that the reconstruction distribution (blue) is centered around the data point.}
    \label{fig:circle}
\end{figure}
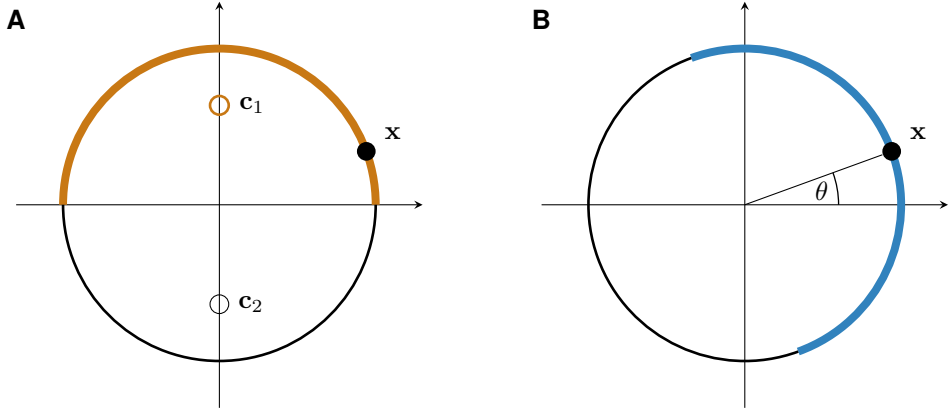

\subsection{Universal quantization}
\label{sec:uq}
Universal quantization \citep{ziv1985universal} is a simple trick which allows us to communicate a uniform sample using bits.
Given $y \in \mathbb{R}$ and a random variable $U$ uniformly distributed over any unit interval, the encoder computes
\begin{align}
    K = \lfloor y + U \rceil.
\end{align}
The decoder receives $K$ and computes the reconstruction
\begin{align}
    \hat Y = K - U.
\end{align}
It can be shown that $\hat Y \sim y + U'$ where $U'$ is uniformly distributed between $-0.5$ and $0.5$. Note that universal quantization requires that the encoder and decoder both have access to $U$. We will adapt this trick to communicate a noisy version of the input angle $\theta$. For the encoder, we choose
\begin{align}
    f(\x, u) = \left\lfloor \frac{\theta}{\pi} + u \right\rceil \mathrm{mod}\ 2.
\end{align}
The decoder receives $k \in \{0, 1\}$ and computes
\begin{align}
    g(k, u) = \begin{pmatrix}\cos\hat\theta \\ \sin\hat\theta\end{pmatrix} 
    \quad\text{where}\quad
    \hat\theta = \pi \left( k - u \right).
\end{align}
The modulo operation is only changing the decoder output by a multiple of 2.
Hence, for a fixed input with angle $\theta$, the reconstruction is distributed as follows:
\begin{align}
    \hat\Theta
    = \pi(k - U)
    = \pi\left( \left\lfloor \frac{\theta}{\pi} + U \right\rceil - 2J - U \right)
    \sim \pi\left( \frac{\theta}{\pi} + U' - 2J \right)
    = \theta + \pi U' - 2\pi J
\end{align}
for some integer valued random variable $J$ which depends on $\theta$ and $U$. Note that we only care about $\theta$ up to multiples of $2\pi$ since $d(\theta, \hat\theta + 2\pi) = d(\theta, \hat\theta)$. Figure~\ref{fig:circle}C illustrates the conditional distribution of reconstructions when using universal quantization. Note that the distribution is centered around the data point, unlike for the deterministic encoder.

The expected distortion is independent of $\theta$ and we have
\begin{align}
    \mathbb{E}[d(\theta, \hat \Theta) \mid \theta]
    = 1 - \int_{\theta - \pi/2}^{\theta + \pi/2} \frac{1}{\pi} \cos(\theta - \hat\theta) \, d\hat \theta
    = 1 - \frac{2}{\pi} < 1 - \frac{4}{\pi^2}.
\end{align}
This distortion is about 38.9\% smaller than with a deterministic encoder. Appendix~D compares the two approaches further at other bit rates.

\section{Discussion}
Stochastic encoders can be much more convenient to optimize than deterministic encoders. 
It is thus important to understand the impact of noise on performance. When only optimizing for distortion, the potential for stochastic encoders to outperform deterministic encoders is limited (Appendix~A). In practice, stochasticity can even degrade the performance of neural encoders \citep{agustsson2020uq}.

On the other hand, here we showed that stochasticity can help significantly when taking into account perceptual quality (measured in terms of the reconstructions' distribution). This justifies using stochastic encoders beyond reasons of convenience and differentiability. In Appendix~B we further show that stochasticity is not enough but that \textit{shared} randomness is needed to achieve an improvement over deterministic encoders in our example. 

Our stochastic encoder relied on universal quantization. Universal quantization is of particular interest because it is easy to implement and uniform noise is commonly used in neural compression. The setting of perfect perceptual quality is also practically relevant since it corresponds to artefact free codecs.

However, here we only considered a simple example in a one-shot setting (that is, encoding a single data point). From rate-distortion theory we know that a deterministic encoder (vector quantization) asymptotically achieves minimal distortion if we encode many data points at once. An important question for future research is whether a similar achievability result holds in the setting of perfect perceptual quality or whether we should also expect noisy encoders to perform significantly better in a wide variety of realistic settings.

\subsubsection*{Acknowledgments}
The authors would like to thank Johannes Ballé for helpful discussions.

\bibliography{references}
\bibliographystyle{iclr2021_conference}

\section*{Appendix A}
Here we show that under mild assumptions deterministic encoders are optimal for a given rate-distortion trade-off. This generalizes a result by \cite{balle2021ntc} showing that a random offset used in dithering can always be replaced by a fixed offset which performs at least as well.

\begin{theorem}
    For every stochastic codec $(f, g) \in \mathcal{F} \times \mathcal{G}$ and trade-off parameter $\lambda \in [0, \infty)$, we can derive a deterministic codec which performs at least as well in terms of the expected rate-distortion trade-off.
\end{theorem}
\begin{proof}
    Let $f \in \mathcal{F}$ be a stochastic encoder and $g \in \mathcal{G}$ be a stochastic decoder. Let $U \sim U([0, 1))$ be a uniform random variable representing
    a shared source of randomness. Further, let 
    \begin{align}
        r: \mathbb{N}_0 \times \mathbb{R} \rightarrow [0, \infty], \quad (\mathbf{x}, u) \mapsto r(\mathbf{x}, u)
    \end{align}
    represent a number of bits, rate, or other coding cost which in general may depend on the state of the noise. Now consider the rate-distortion trade-off for a given trade-off parameter $\lambda \in [0, \infty)$,
    \begin{align}
        \label{eq:tradeoff}
        \ell_\lambda(\mathbf{x}, u) = r(f(\mathbf{x}, u), u) + \lambda d(\mathbf{x}, g(f(\mathbf{x}, u), u)).
    \end{align}
    Under mild conditions on $\ell_\lambda$ (see Tonelli's theorem), we can write for the expected trade-off
    \begin{align}
        l_\lambda = \mathbb{E}[\ell_\lambda(\mathbf{X}, U)] = \mathbb{E}_U[\mathbb{E}_\mathbf{X}[\ell_\lambda(\mathbf{X}, U)]].
    \end{align}
    We know there must exist a $u^* \in [0, 1]$ such that $\mathbb{E}_\X[\ell_\lambda(\mathbf{X}, u^*)] \leq l_\lambda$, otherwise the expected trade-off would be larger. We therefore have
    \begin{align}
        \mathbb{E}_U[\mathbb{E}_\mathbf{X}[\ell_\lambda(\mathbf{X}, U)]] 
        &\geq \mathbb{E}_\mathbf{X}[\ell_\lambda(\mathbf{X}, u^*)] \\
        &= \mathbb{E}_\mathbf{X}[\tilde r(\tilde f(\mathbf{X})) + \lambda d(\mathbf{X}, \tilde g(\tilde f(\mathbf{X})))],
        \label{eq:lagrangian}
    \end{align}
    where
    \begin{align}
        \tilde f(\mathbf{x}) = f(\mathbf{x}, u^*)
    \end{align}
    is a deterministic encoder and $\tilde g$ and $\tilde r$ are defined analogously. That is, for every stochastic codec we can derive a deterministic codec which performs at least as well in terms of the expected rate-distortion trade-off.
\end{proof}

Interestingly, the same cannot be said if we consider the constrained optimization problem to
\begin{align}
    \text{minimize}_{f, g} \quad \mathbb{E}[r(f(\mathbf{X}, U), U)]
    \quad \text{s.t.} \quad \mathbb{E}[d(\mathbf{X}, g(f(\mathbf{X}, U), U))] < D.
\end{align}
Here, stochastic encoders may already perform better as \cite{wagner2021sawbridge} showed in an example. This can happen when the Pareto front of deterministic encoders is non-convex such that its Lagrangian (that is, the trade-off in Eq.~\ref{eq:lagrangian}) can only address a subset of the constrained optimization problems. However, \cite{wagner2021sawbridge} also showed that in these cases a simple randomized interpolation of two deterministic codecs can achieve optimal performance.

\section*{Appendix B}

The following theorem shows that if the encoder and decoder do not have a access to a shared source of randomness (Fig.~\ref{fig:graphs}C), then requiring perfect perceptual quality necessarily leads to a twofold increase in squared error relative to a codec which only minimizes distortion.

\begin{theorem}
    \label{th:double}
    Let $\mathcal{H}_R \subseteq \mathcal{F}_R \times \mathcal{G}$ be the subset of encoders $f$ and decoders $g$ such that $\hatX \sim \X$ where
    \begin{align}
        U_1, U_2 &\sim U([0, 1)), \\
        K &= f(\X, U_1), \\
        \hatX &= g(K, U_2),
    \end{align}
    and $U_1$ and $U_2$ are assumed to be independent. Then
    \begin{align}
        \inf_{f, g \in \mathcal{H}_R} \mathbb{E}[\|\X - \hatX\|^2]
        \geq 2 \inf_{f \in \mathcal{F}_R} \mathbb{E}[\|\X - \mathbb{E}[\X \mid K]\|^2].
    \end{align}
\end{theorem}

\begin{proof}
    We can divide the error into two components,
    \begin{align}
        &\mathbb{E}[\| \X - \hatX \|^2] \\
        =\ &\mathbb{E}[\| \X - \mathbb{E}[\X \mid K] + \mathbb{E}[\X \mid K] - \hatX \|^2] \\
        =\ &\mathbb{E}[\| \X - \mathbb{E}[\X \mid K] \|^2 + \| \mathbb{E}[\X \mid K] - \hatX \|^2
        + 2(X - \mathbb{E}[\X \mid K])^\top (\mathbb{E}[\X \mid K] - \hatX)] \\
        =\ &\mathbb{E}[\| \X - \mathbb{E}[\X \mid K] \|^2] + \mathbb{E}[\| \mathbb{E}[\X \mid K] - \hatX \|^2],
    \end{align}
    since
    \begin{align}
        &\mathbb{E}[(\X - \mathbb{E}[\X \mid K])^\top (\mathbb{E}[\X \mid K] - \hatX)] \\
        =\ &\mathbb{E}_{K}[\mathbb{E}_{\X, \hatX}[(\X - \mathbb{E}[\X \mid K])^\top (\mathbb{E}[\X \mid K] - \hatX) \mid K]] \\
        =\ &\mathbb{E}_K[\mathbb{E}_\X[\X - \mathbb{E}[\X \mid K] \mid K]^\top \mathbb{E}_{\hatX}[\mathbb{E}[\X \mid K] - \hatX \mid K]] \\
        =\ &\mathbb{E}_K[(\mathbb{E}[\X \mid K] - \mathbb{E}[\X \mid K])^\top (\mathbb{E}[\X \mid K] - \mathbb{E}[\hatX \mid K])] \\
        =\ &0.
    \end{align}
    The second equality is true because $\X$ and $\hatX$ are independent conditioned on $K$.
    It follows that
    \begin{align}
        &\inf_{f, g \in \mathcal{H}_R} \mathbb{E}[\| \X - \hatX \|^2] \\
        =\ &\inf_{f, g \in \mathcal{H}_R} \mathbb{E}[\| \X - \mathbb{E}[\X \mid K] \|^2] + \mathbb{E}[\| \mathbb{E}[\X \mid K] - \hatX \|^2] \\
        \geq\ &\inf_{f, g \in \mathcal{H}_R} \mathbb{E}[\| \X - \mathbb{E}[\X \mid K] \|^2] +
        \inf_{f, g \in \mathcal{H}_R} \mathbb{E}[\| \mathbb{E}[\X \mid K] - \hatX \|^2] \\
        \label{eq:twoterms}
        \geq\ &\inf_{f, g \in \mathcal{H}_R} \mathbb{E}[\| \X - \mathbb{E}[\X \mid K] \|^2] +
        \inf_{f, g \in \mathcal{H}_R} \mathbb{E}[\| \mathbb{E}[\hatX \mid K] - \hatX \|^2] \\
        =\ &2\inf_{f \in \mathcal{F}_R} \mathbb{E}[\| \X - \mathbb{E}[\X \mid K] \|^2].
    \end{align}
    Here, the second inequality is true because the conditional expectation minimizes the squared error. To understand the last equality, note that both terms in Eq.~\ref{eq:twoterms} represent the same minimization problem over the same set of joint distributions over $\X$ and $K$ (or $\smash{\hatX}$ and $K$). In the first minimization problem we fix $p(\x)$ and chose $p(k \mid \x)$. In the second minimization problem we equivalently choose both $p(k)$ and $p(\hatx \mid k)$ subject to a constraint on the marginal distribution of $p(\hatx)$.
\end{proof}

A result of \citet[Theorem 2]{blau2019rethinking} could be adapted to prove the achievability of this bound (we assume a fixed rate here where they constrain the mutual information). While the minimum squared error increases at most by a factor of two by adding a \textit{perfect perceptual quality} constraint \citep{blau2019rethinking}, our result shows that we need a shared source of randomness if we want the error to be any lower.

\section*{Appendix C}
Here we show that the deterministic encoder described in Section~\ref{sec:vq} is optimal among deterministic encoders. 
Let $f \in \mathcal{F}_1$ be any stochastic 1-bit encoder.
Futher, let $K = f(\X, U)$. Given $K$, a decoder $g \in \mathcal{G}$ which minimizes the squared error is given by
\begin{align}
    g(k, u) = \cb_k, \quad\text{where}\quad \cb_k = \mathbb{E}[\X \mid K].
\end{align}
Then the following deterministic encoder $f^* \in \mathcal{F}_1$ performs at least as well,
\begin{align}
    \label{eq:fstar}
    f^*(\x, u) = \underset{k \in \{0, 1\}}{\text{argmin}}\, \|\x - \cb_k\|^2.
\end{align}
In other words, when searching for an encoder which minimizes squared error we only need to consider deterministic encoders of the form of Equation~\ref{eq:fstar}.

The codebook vectors $\cb_0$ and $\cb_1$ create two Voronoi regions in $\mathbb{R}^2$ which split the circle into two sections. The squared error corresponding to a uniform distribution over a section of length $t$ can be calculated to be
\begin{align}
    \mathbb{E}[\|\X - \mathbf{c}^*\|^2 \mid \theta \in [a, a + t)] = \frac{t^2 + 2\cos(t) - 2}{t^2}
\end{align}
where
\begin{align}
    \mathbf{c}^* = \mathbb{E}[\X \mid \theta \in [a, a + t)].
\end{align}
Assume the lengths of the Voronoi sections corresponding to $\cb_0$ and $\cb_1$ are $t$ and $2\pi - t$, respectively. Then the average error is given by
\begin{align}
    \mathbb{E}[\|\X - \mathbb{E}[\X \mid K]\|^2]
    &= \frac{t}{2\pi} \frac{t^2 + 2\cos(t) - 2}{t^2} +
    \frac{2\pi - t}{2\pi} \frac{(2\pi - t)^2 + 2\cos(2\pi - t) - 2}{(2\pi - t)^2} \\
    &= \frac{2 + 2\pi t - t^2 + 2\cos(t)}{2\pi t - t^2}.
\end{align}
This function is minimized at $t = \pi$ where
\begin{align}
    \mathbb{E}[\|\X - \mathbb{E}[\X \mid K]\|^2] = 1 - \frac{4}{\pi^2}.
    \end{align}
Using Theorem~\ref{th:double} (Appendix~B), we conclude
\begin{align}
    \mathbb{E}[d(\Theta, \hat \Theta)]
    = \frac{1}{2} \mathbb{E}[\|\X - \hatX\|^2]
    \geq \inf_{f \in \mathcal{F}_1} \mathbb{E}[\|\X - \mathbb{E}[\X \mid K]\|^2]
    = 1 - \frac{4}{\pi^2}.
\end{align}
That is, there is no deterministic encoder which performs better than the one described in Section~\ref{sec:vq}.

\newpage

\section*{Appendix D}

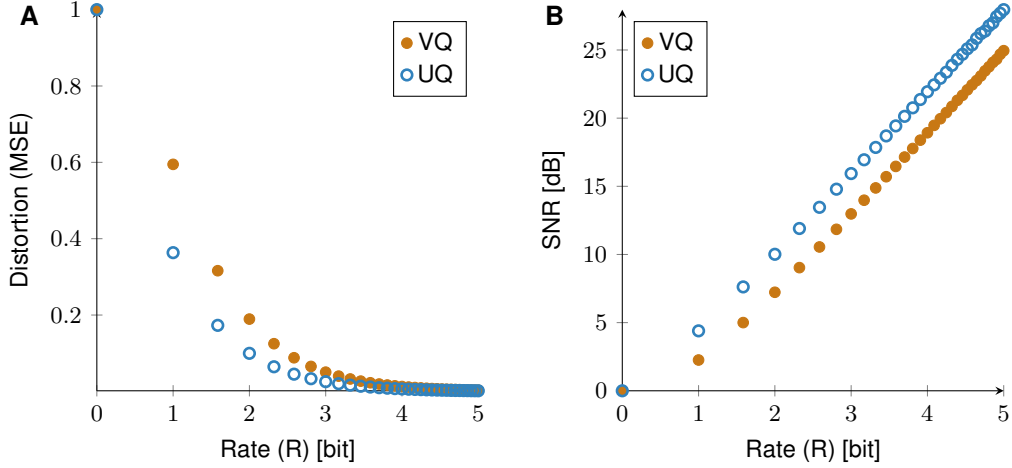
\begin{figure}[h!]
    \centering
    \begin{tikzpicture}
        \node at (-0.9, \textwidth/2.8) {\textbf{\textsf{A}}};
        \begin{axis}[
                width=\textwidth/2.1,
                height=\textwidth/2.1,
                domain=0:32,
                samples=33,
                ylabel={Distortion (MSE)},
                xlabel={Rate (R) [bit]},
                legend entries={VQ, UQ},
                legend pos={north east},
                axis lines=left,
            ]
            
            \addplot[vqcolor, only marks]
            table[
                x expr=log2(\thisrow{n}),
                y expr=\thisrow{vq},
            ] {rdcurve.dat};
            
            \addplot[uqcolor, only marks, mark=o, line width=1pt]
            table[
                x expr=log2(\thisrow{n}),
                y expr=\thisrow{uq},
            ] {rdcurve.dat};
        \end{axis}

        \begin{scope}[xshift=\textwidth/2]
            \node at (-0.9, \textwidth/2.8) {\textbf{\textsf{B}}};
            \begin{axis}[
                    width=\textwidth/2.1,
                    height=\textwidth/2.1,
                    domain=0:32,
                    samples=33,
                    ylabel={SNR [dB]},
                    xlabel={Rate (R) [bit]},
                    legend entries={VQ, UQ},
                    legend pos={north west},
                    axis lines=left,
                ]
    
                \addplot[vqcolor, only marks]
                table[
                    x expr=log2(\thisrow{n}),
                    y expr=-10 * log10(\thisrow{vq}),
                ] {rdcurve.dat};
                
                \addplot[uqcolor, only marks, mark=o, line width=1pt]
                table[
                    x expr=log2(\thisrow{n}),
                    y expr=-10 * log10(\thisrow{uq}),
                ] {rdcurve.dat};
            \end{axis}
        \end{scope}
    \end{tikzpicture}
    \caption{Comparison of a deterministic (VQ) and a stochastic fixed-rate encoder on the circular uniform distribution. \textbf{A}: The mean-squared error as a function of the bit rate. \textbf{B}: The signal-to-noise ratio as a function of the bit rate.}
    \label{fig:rdcurve}
\end{figure}

We showed that the optimal deterministic encoder at target bit rate of 1 is a vector quantizer (VQ) followed by stochastic decoding (Appendix~C). Without proof, here we illustrate the performance of this approach and compare it to universal quantization (UQ) at other rates. If $R = \log_2 N$ is the target bit rate, then we have for the distortion:
\begin{align}
    \mathbb{E}_\text{VQ}[d(\Theta, \hat\Theta)] &= 1 - \left(\frac{N \sin(\pi/N)}{\pi}\right)^2, \\
    \mathbb{E}_\text{UQ}[d(\Theta, \hat\Theta)] &= 1 - \frac{N \sin(\pi/N)}{\pi}.
\end{align}
Figure~\ref{fig:rdcurve} visualizes both the average distortion $D$ and the signal-to-noise ratio $-10 \log_{10} D$ (SNR). As we increase the bit rate, the distortion converges to zero in both cases. On the other hand, the gap between deterministic and stochastic encoders even increases if measured as a difference in SNR.

\end{document}

%% file: macros.tex
\newcommand{\x}{{\mathbf{x}}}
\newcommand{\hatx}{{\mathbf{\hat x}}}
\newcommand{\X}{{\mathbf{X}}}
\newcommand{\hatX}{{\mathbf{\hat X}}}
\newcommand{\cb}{{\mathbf{c}}}